\documentclass[lnicst]{svmultln}
\usepackage{makeidx}  

\usepackage{amsmath}
\usepackage{algorithm}
\usepackage{algpseudocode}
\algnotext{EndFor}
\algnotext{EndIf}

\usepackage{tikz}
\usetikzlibrary{arrows}
\usepackage{color}
\usepackage{subfigure}

\usepackage[textwidth=14.0cm,textheight=19.3cm]{geometry}

\def\lengtharrow{30mm}
\def\lengtharrowdouble{39mm}
\newcommand{\boitegauche}[1]{\begin{minipage}[t][][t]{45mm}\flushright#1~\end{minipage}}
\newcommand{\boitedroite}[1]{\begin{minipage}[t][][t]{45mm}\flushleft#1~\end{minipage}}
\newcommand{\boitecentre}[1]{\begin{minipage}[t][][t]{\lengtharrow}{\hbox to \lengtharrow{\hfill#1\hfill}}\end{minipage}}
\newcommand{\flechegauche}[1]{$\stackrel{\text{#1}}{\hbox to \lengtharrow{\leftarrowfill}}$}
\newcommand{\flechedroite}[1]{$\stackrel{\text{#1}}{\hbox to \lengtharrow{\rightarrowfill}}$}
\newcommand{\flechedouble}[1]{$\stackrel{\text{#1}}{\xleftrightarrow{\hbox to \lengtharrowdouble{}}}$}
\newcommand{\pslow}{\hspace{-0mm}\left.}
\newcommand{\pfast}{\hspace{-4pt}\left[}
\newcommand{\pinit}{\hspace{-0mm}\left.}
\newcommand{\pverif}{\hspace{-0mm}\left.}
\newcommand{\interphase}{\vspace{-1mm}}
\def\rowh{.8} 
\newcommand\inR{\in_\text{\sf R}}

\newcommand\tmax{t_\text{max}}

\begin{document}
	\mainmatter              
	%

\title{Complexity of distance fraud attacks in graph-based distance bounding}
					%
	%
	\author{Rolando Trujillo-Rasua}
	%
	%
	%
	\institute{University of Luxembourg, SnT\\
		\email{rolando.trujillo@uni.lu}}
	
	\maketitle              
	
	\begin{abstract}        
		\emph{Distance bounding} (DB) emerged as a countermeasure to the 
		so-called \emph{relay attack}, which affects several technologies such 
		as RFID, NFC, Bluetooth, and Ad-hoc networks. A prominent family of DB 
		protocols are those based on graphs, which were introduced in 2010 to 
		resist both mafia and distance frauds. The security analysis in terms 
		of distance fraud is performed by considering an adversary that, given 
		a vertex labeled graph $G = (V, E)$ and a vertex $v \in V$, is able to 
		find the most frequent $n$-long sequence  in $G$ starting from $v$ 
		(MFS 
		problem). However, to the best of our knowledge, it is still an 
		open question whether the distance fraud security can be computed 
		considering the aforementioned adversarial model. Our first 
		contribution is a 
		proof that the MFS problem is NP-Hard even when the graph is 
		constrained to meet the requirements of a graph-based DB protocol. 
		Although this result does not invalidate the model, it does 
		suggest that a \emph{too-strong} adversary is 
		perhaps being considered (\emph{i.e.}, in practice, graph-based DB 
		protocols 
		might resist distance fraud better than the security 
		model suggests.) Our second contribution is an algorithm addressing 
		the distance fraud security of the tree-based approach due to Avoine 
		and Tchamkerten. The novel algorithm improves the computational 
		complexity 
		$O(2^{2^n+n})$ of the naive approach to $O(2^{2n}n)$ where $n$ is the number of rounds.
	\end{abstract}

	\section{Introduction}

	Let us consider a little girl willing to compete with two chess 
	grandmasters, say Fischer and Spassky. She agrees with both on playing by 
	post and manages to use opposite-colored pieces in the games. 
	Once the 
	little girl receives Fisher's move she simply forwards it to Spassky and 
	vice versa. As a result, she wins one game or draws both even though she 
	might know nothing about chess. This problem, known as the \emph{chess 
	grandmaster problem}, was introduced by Conway in 
	1976~\cite{citeulike:1223195} and informally describes how relay attacks 
	work. 
	
	In a relay attack, an adversary acts as a \emph{passive} man-in-the-middle attacker relaying messages between the prover and the verifier during an authentication protocol. In case the adversary is \emph{active}, the attack is known as \emph{mafia fraud}~\cite{Desmedt:1987:SUS:646752.704723} and succeeds if the prover and the verifier complete the authentication protocol without noticing the presence of the adversary.
	
	With the widespread deployment of contactless technologies in recent 
	years, mafia fraud has re-emerged as a serious security threat for 
	authentication schemes. Radio Frequency IDentification (RFID), Near Field 
	Communication (NFC), and Passive Keyless Entry and Start Systems in Modern 
	Cars, have been proven to be vulnerable to mafia 
	fraud~\cite{Francis:2010:PNP:1926325.1926331,Hancke:2008:ATD:1352533.1352566}.
	 Other contactless technologies such as smartcards and e-voting are also 
	threatened by this attack~\cite{DrimerM-2007-usenix,OrenW-2009-eprint}.
	
	The most promising countermeasure to thwart mafia fraud is \emph{distance 
	bounding} (DB)~\cite{Brands:1994:DP:188307.188361}, that is, an 
	authentication protocol where time-critical sessions allow to compute an 
	upper bound of the distance between the prover and the verifier. However, 
	this type of protocols is vulnerable to another type of fraud, the 
	\emph{distance fraud}~\cite{Brands:1994:DP:188307.188361}. Contrary to 
	mafia fraud, distance fraud is performed by a legitimate prover, who aims 
	to authenticate beyond the expected and allowed distance. 
	
	In 2010, graph-based DB protocols aimed at being resistant to both mafia 
	and distance frauds were 
	introduced~\cite{Trujillo-Rasua:2010:PDP:1926325.1926352}. This type of 
	protocols is flexible in the sense that different graph structures can be 
	used so as to balance memory requirements and security properties. 
	However, neither the graph-based approach 
	in~\cite{DBLP:conf/isw/AvoineT09} nor the one 
	in~\cite{Trujillo-Rasua:2010:PDP:1926325.1926352} have computed their 
	actual distance fraud security. Indeed, this analysis was left as an 
	open problem in~\cite{Trujillo-Rasua:2010:PDP:1926325.1926352}.

	\noindent\textbf{Contributions.} 	In this article we address the open 
	problem of computing the distance fraud resistance of graph-based DB 
	protocols. We first reformulate the security model provided 
	in~\cite{Trujillo-Rasua:2010:PDP:1926325.1926352} and define it in terms 
	of, to the best of our knowledge, two new problems in Graph Theory. 
	\emph{The Most Frequent 
		Sequence problem} (MFS problem) and its simplified version \emph{the 
		Binary Most Frequent Sequence problem} (Binary 
	MFS problem).
	
	We then provide a polynomial-time reduction of the Satisfiability problem 
	(SAT) to the Binary MFS problem, proving that both the Binary MFS and the 
	MFS problems are NP-Hard. This result suggests that a \emph{too-strong} 
	adversary is perhaps being considered by the security model, unless $P = 
	NP$. However, the 
	implications of our reduction goes beyond that. 
	It also provides a clue of how to design graph-based DB protocols 
	resistant to distance fraud. 
	
	Our next contribution is a novel algorithm to compute the distance fraud 
	resistance of the tree-based DB protocol proposed by Avoine and 
	Tchamkerten~\cite{DBLP:conf/isw/AvoineT09}. Our algorithm significantly 
	reduces the time complexity of the naive approach from $O(2^{2^n+n})$ to 
	$O(2^{2n}n)$ where $n$ is the number of rounds. This paves the way for a fair 
	comparison of graph-based proposals with other state-of-the-art DB protocols.
	
	\noindent\textbf{Organization.}	The rest of this article is organized as follows. 
	Section~\ref{sec_preliminaries} introduces graph-based DB 
	protocols and the new problems Binary MFS and MFS. Related 
	works close to the MFS problem are reviewed in 
	Section~\ref{sec_preliminaries} as well. Section~\ref{sec_np} contains proofs on the 
	hardness of the Binary MFS problem. The algorithm for computing the distance 
	fraud resistance of the tree-based DB protocol is described and 
	analyzed in Section~\ref{sec_algorithm}. The discussion and conclusions are drawn in Section~\ref{sec_conclusions}.
	
	\section{Preliminaries}\label{sec_preliminaries}
	
%
	
	\subsection{Graph-based distance bounding protocols}
	
	Graph-based DB protocols were introduced 
	in~\cite{Trujillo-Rasua:2010:PDP:1926325.1926352} aimed at resisting both 
	mafia and distance frauds, yet requiring low memory to be implemented. The 
	idea is to define a 
	digraph $G = (V, E)$ and a starting vertex $v \in V$. Then, a 
	challenge-response protocol (see Figure~\ref{fig_graph_protocol}) is executed 
	where the challenges define a walk in $G$ according to an edge labeling 
	function $\ell_E: E \rightarrow \{0, 1\}$ and the 
	responses are stored on the vertices according to a 
	vertex labeling function $\ell_V: V \rightarrow \{0, 1\}$.
	
	{\renewcommand{\arraystretch}{\rowh}
		\begin{algorithm}[!htb]
			\begin{center}
				\[\pinit%
				\begin{tabular}{ccc}
				\boitegauche{\textbf{Verifier}} & \boitecentre{} & \boitedroite{\textbf{Prover}} \\
				\boitegauche{(Secret $x$)} && \boitedroite{(Secret $x$)} \\
				\boitegauche{(Digraph $G = (V, E)$)} && \boitedroite{(Digraph $G = (V, E)$)} \\
				\boitegauche{(Starting vertex $v \in V$)} && \boitedroite{(Starting vertex $v \in V$)} \\
				\end{tabular}
				\right.\]
				\interphase
				\[\pslow%
				\begin{tabular}{rcl}
				\boitegauche{Pick a random $N_V$} & &  \boitedroite{Pick a random $N_P$}\\
				\boitegauche{}& \flechedroite{$N_V$} & \boitedroite{}\\
				\boitegauche{}& \flechegauche{$N_P$} & \boitedroite{}\\
				\boitegauche{
					\begin{tabular}{l}
					On input $PRF(K,N_{P},N_{V})$\\
					Create $\ell_V$ and $\ell_E$\\
					$v_0' = v$\\
					\end{tabular}} &
				& \boitedroite{%
					\begin{tabular}{l}
					On input $PRF(K,N_{P},N_{V})$\\
					Create $\ell_V$ and $\ell_E$\\
					$v_0 = v$\\
					\end{tabular}}
				\end{tabular}
				\right.\]
				\interphase
				\[\pfast
				\begin{tabular}{rcl}
				\boitegauche{}& \textbf{for }$i = 1$ \textbf{ to } $n$& \boitedroite{}\\
				\boitegauche{Pick $c_i \inR \{0,1\}$} & \boitecentre{}&\boitedroite{}\\
				\boitegauche{Start Timer} & \flechedroite{$c_i$} &\\
				& & \boitedroite{Look for $v_i$ such that }\\
				\boitegauche{Stop Timer} & \flechegauche{$\ell_V(v_i)$} &\boitedroite{$\ell_E(v_{i-1}, v_i) = c_i$}\\
				\boitegauche{Look for $v_i'$ such that $\ell_E(v_{i-1}', v_i') = c_i$}& & \\
				\end{tabular}
				\right.\]
				\interphase
				\[\pverif
				\begin{tabular}{lcc}
				\boitegauche{Check that $\Delta t_i \leq \tmax$ and $\ell_V(v_i') = 
					\ell_V(v_i)$}&\boitecentre{}&\boitedroite{}\\
				\end{tabular}
				\right.\]
			\end{center}
			\caption{Graph-based distance bounding protocol}\label{fig_graph_protocol}
		\end{algorithm}}

		As shown by Figure~\ref{fig_graph_protocol}, prover and verifier exchange two 
		nonces and use a pseudo-random function ($PRF(.)$) with a shared 
		private 
		key to compute two labeling functions $\ell_V: V \rightarrow \sum$ and 
		$\ell_E: E \rightarrow \sum$ where $\sum = \{0, 1\}$. For 
		$\ell_E$ it must hold that $\ell_E(u, v) \neq \ell_E(u, w)$ for every 
		pair of different
		edges $(u, v)$ and $(u, w)$ in $E$. By contrast, $\ell_V$ is chosen randomly. 
		After labeling the graph, $n$ rounds of time-critical sessions are executed. 
		At the $i$th round, the verifier sends the binary challenge $c_i$. 
		Then, the prover 
		answers with $\ell_V(v_i)$ where $v_i$ holds that $(v_{i-1}, v_i) \in E$ and 
		$\ell_E(v_{i-1}, v_i) = c_i$. Note that $v_0$ is the starting vertex 
		$v$. At 
		the end of the $n$ time-critical sessions, the verifier checks all prover's 
		responses ($\ell_V(v_i)$) and the round-trip-times ($\Delta t_i$), 
		which should be below some threshold $\tmax$. Intuitively, the lower $\tmax$ 
		the closer the prover to the verifier is expected to be. 	
		
		Two graph-based DB protocols exist; the tree-based approach~\cite{DBLP:conf/isw/AvoineT09} and the Poulidor protocol~\cite{Trujillo-Rasua:2010:PDP:1926325.1926352}. As suggested by its name, the former uses a tree of depth $n$ and $2^{n+1}-1$ nodes (see Figure~\ref{fig_tree}). By contrast, Poulidor uses a graph structure with $2n$ nodes only in order to reduce memory requirements (see Figure~\ref{fig_poulidor}). Both have proven to resist mafia fraud better than other DB protocols such as~\cite{Hancke:2005:RDB:1128018.1128472,DBLP:conf/cans/KimA09}. Its resistance to distance fraud, however, is still an open problem. 
		
		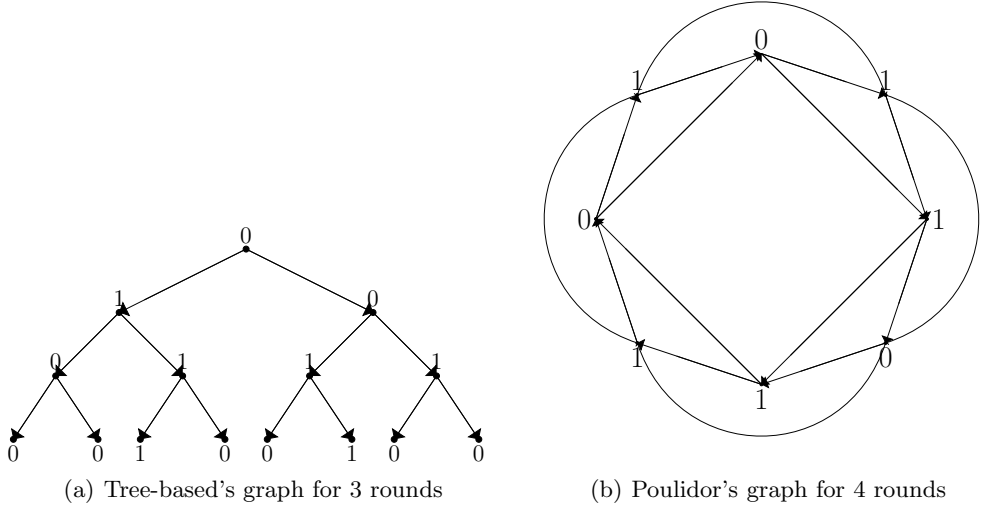
\begin{figure}
			\centering
			\subfigure[Tree-based's graph for $3$ rounds]{
				\scalebox{0.75}{\begin{tikzpicture} [scale = 1.5]

        \filldraw [black] (0,1.5) circle (1pt)  node[anchor = south] {\large{$0$}}; 

        \filldraw [black] (-1.5,0.75) circle (1pt) node[anchor = south] {\large{$1$}}; 
        \filldraw [black] (1.5,0.75) circle (1pt) node[anchor = south] {\large{$0$}}; 

        \draw[-triangle 90,postaction={draw, shorten >=0.5cm, -}] (0,1.5) -- (-1.5,0.75) ; 
        \draw[-triangle 90,postaction={draw, shorten >=1mm, -}] (0,1.5) -- (1.5,0.75) ;

        \filldraw [black] (-2.25,0) circle (1pt) node[anchor = south] {\large{$0$}}; 
        \filldraw [black] (-0.75,0) circle (1pt) node[anchor = south] {\large{$1$}}; 
        \filldraw [black] (0.75,0) circle (1pt) node[anchor = south] {\large{$1$}}; 
        \filldraw [black] (2.25,0) circle (1pt) node[anchor = south] {\large{$1$}}; 

        \draw[-triangle 90,postaction={draw, shorten >=1mm, -}] (-1.5,0.75) -- (-2.25,0);
       \draw[-triangle 90,postaction={draw, shorten >=1mm, -}] (-1.5,0.75) -- (-0.75,0);
        \draw[-triangle 90,postaction={draw, shorten >=1mm, -}] (1.5,0.75) -- (0.75,0);
        \draw[-triangle 90,postaction={draw, shorten >=1mm, -}] (1.5,0.75) -- (2.25,0);

        \filldraw [black] (-2.75, -0.75) circle (1pt) node[anchor = north] {\large{$0$}}; 
        \filldraw [black] (-1.75,-0.75) circle (1pt) node[anchor = north] {\large{$0$}}; 
        \filldraw [black] (-1.25, -0.75) circle (1pt) node[anchor = north] {\large{$1$}}; 
        \filldraw [black] (-0.25,-0.75) circle (1pt) node[anchor = north] {\large{$0$}}; 
        \filldraw [black] (0.25, -0.75) circle (1pt) node[anchor = north] {\large{$0$}}; 
        \filldraw [black] (1.25,-0.75) circle (1pt) node[anchor = north] {\large{$1$}}; 
        \filldraw [black] (1.75, -0.75) circle (1pt) node[anchor = north] {\large{$0$}}; 
        \filldraw [black] (2.75,-0.75) circle (1pt) node[anchor = north] {\large{$0$}}; 

        \draw[-triangle 90,postaction={draw, shorten >=1mm, -}] (-2.25,0) -- (-2.75, -0.75);
       \draw[-triangle 90,postaction={draw, shorten >=1mm, -}] (-2.25,0) -- (-1.75,-0.75);
       \draw[-triangle 90,postaction={draw, shorten >=1mm, -}] (-0.75,0) -- (-1.25, -0.75);
        \draw[-triangle 90,postaction={draw, shorten >=1mm, -}] (-0.75,0) -- (-0.25,-0.75);
        \draw[-triangle 90,postaction={draw, shorten >=1mm, -}] (0.75,0) -- (0.25, -0.75);
       \draw[-triangle 90,postaction={draw, shorten >=1mm, -}] (0.75,0) -- (1.25,-0.75);
        \draw[-triangle 90,postaction={draw, shorten >=1mm, -}] (2.25,0) --  (1.75, -0.75);
        \draw[-triangle 90,postaction={draw, shorten >=1mm, -}] (2.25,0) -- (2.75,-0.75);

\end{tikzpicture}}
				\label{fig_tree}
			}
			\subfigure[Poulidor's graph for $4$ rounds]{
				\scalebox{0.55}{\begin{tikzpicture}

\filldraw [black] (0,4) circle (0.1pt) node[anchor = south] {\huge{$0$}};
\filldraw [black] (0,-4) circle (1pt) node[anchor = north] {\huge{$1$}};
\filldraw [black] (4,0) circle (1pt) node[anchor = west] {\huge{$1$}};
\filldraw [black] (-4,0) circle (1pt) node[anchor = east] {\huge{$0$}};
\filldraw [black] (3,3) circle (1pt) node[anchor = south] {\huge{$1$}};
\filldraw [black] (3,-3) circle (1pt) node[anchor = north] {\huge{$0$}};
\filldraw [black] (-3,3) circle (1pt) node[anchor = south] {\huge{$1$}};
\filldraw [black] (-3,-3) circle (1pt) node[anchor = north] {\huge{$1$}};

\draw[-triangle 90,postaction={draw, shorten >=1cm, -}] (0, 4)  --  (4, 0) ;
\draw[-triangle 90,postaction={draw, shorten >=1cm, -}]  (0, 4)  -- (3, 3) ;

\draw[-triangle 90,postaction={draw, shorten >=1cm, -}] (3, 3)  --  (4, 0) ;
\draw[draw,-to] (3,3) .. controls (6,2) and (6,-2)  .. (3,-3) ;

\draw[-triangle 90,postaction={draw, shorten >=1cm, -}]  (4, 0)  --   (0, -4) ;
\draw[-triangle 90,postaction={draw, shorten >=1cm, -}]  (4, 0)  --  (3, -3) ;

\draw[-triangle 90,postaction={draw, shorten >=1cm, -}]  (3, -3)  --  (0, -4) ;
\draw[draw,-to]  (3,-3) .. controls (2,-6) and (-2,-6)  .. (-3,-3) ;

\draw[-triangle 90,postaction={draw, shorten >=1cm, -}]   (0, -4)  --   (-4, 0) ;
\draw[-triangle 90,postaction={draw, shorten >=1cm, -}] (0, -4)  --   (-3, -3) ;

\draw[-triangle 90,postaction={draw, shorten >=1cm, -}] (-3, -3)  --  (-4, 0) ;
\draw[draw,-to]  (-3,-3) .. controls (-6,-2) and (-6,2)  .. (-3,3) ;

\draw[-triangle 90,postaction={draw, shorten >=1cm, -}]  (-4, 0) --  (0, 4) ;
\draw[-triangle 90,postaction={draw, shorten >=1cm, -}]  (-4, 0)  --  (-3, 3) ;

\draw[-triangle 90,postaction={draw, shorten >=1cm, -}] (-3, 3)  --  (0, 4) ;
\draw[draw,-to]  (-3,3) .. controls (-2,6) and (2,6)  .. (3,3) ;

\end{tikzpicture}}
				\label{fig_poulidor}
			}
			\caption{Graph structures used by the tree-based and Poulidor approaches. Vertices of both graphs have been randomly labeled.}
		\end{figure}
	
		There exist other DB protocols, more computationally demanding, based on signatures and/or a final extra slow phase~\cite{Brands:1994:DP:188307.188361,SingeleeP-2007-esas}. Others simply could be plugged into most DB protocols such as~\cite{Munilla:2008:DBP:1455665.1455675,XinYTHC-2011-imccc}. The interested reader could refer to~\cite{AvoineBKLM-2011-jcs} for more details.

		\subsection{Distance fraud security analysis}

		
		The security analysis in terms of distance fraud is usually performed within a well-known framework proposed by Avoine et al.~\cite{AvoineBKLM-2011-jcs}. In this framework, a distance fraud adversary uses the \emph{early-reply strategy} to defeat the DB protocol. This strategy consists on sending the bits answer in advance (\emph{i.e.,} before receiving the challenges.) Doing so, the adversary simulates to be closer than really is, and its success probability is lower-bounded by $1/2^n$.
		
		In~\cite{Trujillo-Rasua:2010:PDP:1926325.1926352}, the best early-reply strategy against what they called a \emph{family} of DB protocols is defined. This \emph{family} includes graph-based DB protocols. However, their definition is too generic to be used for simply analyzing graph-based DB protocols. Therefore, we reformulate it here in terms of a new problem in Graph Theory. The problem is named Binary MFS problem (See Definition~\ref{def_binary_mfs_problem}) and is based on its more general version MFS problem (see Definition~\ref{def_mfs_problem}). 
		
		\begin{definition}[The most frequent sequence problem (MFS problem)]\label{def_mfs_problem}
			Let $G = (V, E)$ be a vertex-labeled digraph where $\Sigma$ and $\ell: V 
			\rightarrow \Sigma$ are the set of vertex labels and the labeling function 
			respectively. For a label sequence $t = t_1t_2...t_k$, $occ_{ v}^G(t)$ denotes 
			the number of walks $v_{1}v_{2}...v_{k}$ in $G$ such that $v_{1} = v$ and 
			$\forall i \in \{1, ..., k\} \ (\ell(v_{i}) = t_i)$. The MFS 
			problem 
			consists on 
			finding, given the triple $(G, v, k)$, the \emph{most frequent sequence} of 
			size $k$ defined as $\arg\max_{t\in \Sigma^{k}}(occ_{ v}^G(t))$.
		\end{definition}
		
		\begin{definition}[Binary MFS problem]\label{def_binary_mfs_problem}
			The Binary MFS problem is an MFS problem  where $G$ is constrained to 
			use a binary vertex set ($\sum = \{0, 1\}$) and the out-degree of every vertex 
			should be at most $2$.
		\end{definition}
		
		\noindent \textbf{Example.} Either the graph in Figure~\ref{fig_tree} or the one in Figure~\ref{fig_poulidor} can be the input of the Binary MFS problem. Assuming $k = 4$ and the starting vertex as the top one, the most frequent sequences for the tree in Figure~\ref{fig_tree} is $0010$ (occurs $3$ times), and for the graph in Figure~\ref{fig_poulidor} is $0101$ (occurs $4$ times). 
		
		To successfully apply a distance fraud attack against a graph-based DB 
		protocol with $n$ time-critical sessions, the best adversary's 
		strategy consists of: i) solving the 
		Binary MFS problem defined by the triple $(G, v, n+1)$ and finding the most frequent sequence $t_0t_1...t_n$, ii) sending $t_1...t_n$ in advance to the verifier as the 
		responses to the $n$ verifier's challenges. By this strategy, the 
		adversary's success probability is maximized to $\frac{occ_{ 
		v}^G(t)}{2^n}$~\cite{Trujillo-Rasua:2010:PDP:1926325.1926352}. Coming 
		back to the previous example, the adversary success probability of the 
		tree-based approach defined by Figure~\ref{fig_tree} is $3/8$, which 
		is higher than the expected lower bound $1/8$.

		\begin{definition}[Distance fraud success probability]\label{def_distance_fraud}
			Let $\prod$ be a graph-based DB protocol with $n$ time-critical 
			sessions that uses the vertex-labeled digraph $G = (V, E)$ and $v \in V$ as 
			the starting vertex. Let $M_{G, v, n}$ be a random variable on the sample 
			space of all labeling functions $\ell: V \rightarrow \Sigma$ that outputs the 
			maximum value $\max(occ_{ v}^G(t))$ where $t \in \{0,1\}^{n+1}$. The 
			\emph{distance fraud success probability} of an adversary against $\prod$ is 
			defined as $\frac{E(M_{G, v, n})}{2^n}$ where $E(M_{G, v, n})$ represents the 
			expectation of the random variable $M_{G, v, n}$.
		\end{definition}

		 Note that, following the design of graph-based DB protocols, Definition~\ref{def_distance_fraud} considers that $G$ is randomly labeled at each execution of the protocol. 
		 
		 To the best of our knowledge, computing distance fraud 
		security according to Definition~\ref{def_distance_fraud} has been only addressed in 
		its seminal work~\cite{Trujillo-Rasua:2010:PDP:1926325.1926352}. Apparently, 
		the problem is one of those problems that remain intractable even if $P = NP$ 
		because all, or almost all, the labeling functions should be considered. For 
		this reason, an upper bound was proposed 
		in~\cite{Trujillo-Rasua:2010:PDP:1926325.1926352} and the exact distance fraud 
		security was left as an open problem. We have shown that this problem depends 
		on a problem named the MFS problem and, in particular, on the Binary MFS 
		problem. Below, we review some work related to them.
		
		\subsection{Review on frequent sequences problems}
		
		Sequential Pattern Mining is a well-studied field introduced by Agrawal and 
		Srikant~\cite{DBLP:conf/icde/AgrawalS95} in 1995. Given a databases of 
		transactions (\emph{e.g.}, customer transactions, medical records, web sessions, 
		etc.) the problem consists on discovering all the sequential patterns with 
		some minimum support. The support of a pattern is defined as the number of 
		data-sequence within the database that are contained in the pattern. 
		
		The sequential pattern mining problem is $\#P$-complete~\cite{Yang:2004:CMM:1014052.1014091} and several variants of it exist. For instance, Mannila et al. say that two events are connected if they are close enough in terms of time~\cite{Mannila:1997:DFE:593416.593449}. They define an \emph{episode} as a collection of connected events and the problem is to find frequently occurring episodes in a sequence. A simpler variant, known as the \emph{most common subsequent problem}, was introduced by Campagna and Pagh~\cite{Campagna:2010:FFP:1933307.1934549}. The most common subsequent problem does not consider time-stamped events. Instead, it aims to find all the label sequences in a vertex-labeled acyclic graph that appear more often. Other variants have arisen from complex applications namely, telecommunication, market analysis, and DNA research. We refer the reader to~\cite{mining_survey} for an extensive survey on this subject.
		
		Frequent paths on a graph have also been used to define Kernel functions~\cite{enumerating,Akutsu20121416}. Kernel functions has applicability in chemoinformatics and bioinformatics where objects are mapped to a feature space. In this case, the feature space representation is the number of occurrences of vertex-labeled paths and the problem is to infer the graph from such a feature vector. This problem has been proven to be NP-Hard even for trees of bounded degree~\cite{Akutsu20121416}.
		
		It can be seen that the MFS problem is different to the sequential pattern mining problems and its nature is obviously different to the one of Kernel methods. On one hand, sequential pattern mining is an enumeration problem while MFS is just a search problem. On the other hand, the MFS problem requires all walks to begin from a given vertex and the size of the sequences should be equal. As in~\cite{Campagna:2010:FFP:1933307.1934549}, the time dimension is not considered.
		
	\section{On the hardness of the Binary MFS problem}\label{sec_np}

	Binary MFS is a search problem that looks for the most frequent sequence of 
	length $k$ in a vertex-labeled digraph $G$ starting from a given vertex $v$ 
	(see 
	Definition~\ref{def_binary_mfs_problem}). Intuitively, all or almost all the 
	walks in $G$ starting from $v$ should be analyzed in order to find such a 
	sequence, which means that Binary MFS might not be in the complexity class 
	$P$. However, we cannot even state that Binary MFS is in $NP-P$ since it is 
	not trivial how to \emph{check} a solution in polynomial time. Nevertheless, 
	we prove in this Section (see Theorem~\ref{theo_np}) that the general 
	Boolean 
	Satisfiability problem (SAT) reduces to Binary MFS. Therefore, Binary MFS can 
	be considered NP-Hard even thought it may not even be in 
	$NP$~\cite{Garey:1979:CIG:578533}.
	
	\begin{definition}[SAT]\label{def_sat}
		Let $x = (x_1, x_2, ..., x_n)$ be a set of boolean variables and $c_1(x), 
		c_2(x), ..., c_m(x)$ be a set of clauses where $c_i(x)$ is a disjunction of 
		literals. The Boolean satisfiability problem (SAT) consists on deciding whether there exists an 
		assignment for the boolean variables $x$ such that the function $f_{SAT} = 
		c_1(x) \wedge c_2(x) \wedge ... \wedge c_m(x) = 1$.
	\end{definition}
	
	Algorithm~\ref{alg_reduction} shows our reduction from SAT to an instance of 
	the Binary MFS problem. First, it creates a binary tree $T$ of depth $\lceil 
	\log m \rceil$ with $m$ leafs $c_1$, $c_2$, $...$, $c_m$\footnote{The leafs 
		are intentionally labeled by using the same clause names.}, and a graph $G' = 
	(V', E')$ where $V' = \{u_0^2, v_0^2, ..., u_0^n, v_0^n\}$ 
	and $E' = \{(x, y) | \exists k \in \{2, ..., n-1\} (x \in \{u_0^k, v_0^k\} 
	\wedge y \in \{u_0^{k+1}, v_0^{k+1}\})\}$. The graph $G = (V, E)$ is 
	initialized with the two connected components $T$ and $G'$. In addition, 
	$V$ is increased with the vertices $u_i^j$ and $v_i^j$ where $i \in \{1, 
	..., m\}$ and $j \in \{1, ...,n\}$. Then, for each clause $c_i$ and each 
	variable $x_j$ ($j < n$), the vertex $u_i^j$ is connected with $u_0^{j+1}$ 
	and $v_0^{j+1}$ if $x_j \in c_i(x)$, with $u_i^{j+1}$ and $v_i^{j+1}$ 
	otherwise. Similarly, the vertex $v_i^j$ is connected with $u_0^{j+1}$ and 
	$v_0^{j+1}$ if $\neg x_j \in c_i(x)$, with $u_i^{j+1}$ and $v_i^{j+1}$ 
	otherwise. Finally, for every $i \in \{1, ..., m\}$: i) the vertices 
	$u_i^n$ and $v_i^{n}$ are removed together with their incident edges, ii) 
	the edges $(u_{i}^{n-1}, u_{0}^{n})$ and $(v_{i}^{n-1}, u_{0}^{n})$ are 
	added if $x_n \in c_i(x)$, iii) if $\neg x_n \in c_i(x)$ the added edges 
	are $(u_{i}^{n-1}, v_{0}^{n})$ and $(v_{i}^{n-1}, v_{0}^{n})$. The 
	vertex-label function is simply defined as a function that outputs $0$ on input $v_i^j$ for every $i \in \{0, 1, ..., m\}$ and $j \in \{1, ..., n\}$, 
	outputs $1$ otherwise.

		\begin{algorithm}
			\caption{Reduction from the SAT problem} \label{alg_reduction}
			\begin{algorithmic}[1]
				\Require A SAT instance where $x = (x_1, x_2, ..., x_n)$ are the boolean 
				variables and $c_1(x), c_2(x), ..., c_m(x)$ are the set of clauses 
				\State Let $G = (V, E)$ be a digraph with just one vertex named
				\emph{root}
				\State From the root, a directed binary tree with $m$ leafs is created 
				such that all the leafs are at the same depth. The leaf vertices are 
				denoted as $c_1, c_2, ..., c_m$
				\State Let $G' = (V', E')$ be a digraph where $V' = \{u_0^2, v_0^2, ..., u_0^n, v_0^n\}$ 
				and $E' = \{(x, y) | \exists k \in \{2, ..., n-1\} (x \in \{u_0^k, v_0^k\} 
				\wedge y \in \{u_0^{k+1}, v_0^{k+1}\})\}$
				\State Set $G = G \cup G'$
				\ForAll {vertex $c_i$} 
				\State Set $V = V \cup \{u_i^1, v_i^1, u_i^{2}, v_i^{2}, ..., 
				u_i^n, v_i^n\}$
				\State Set $E = E \cup \{(c_i, u_i^1), (c_i, v_i^1)\}$
				\ForAll {$j \in \{1, 2, ..., n-1\}$}
				\If {$x_{j} \in c_i(x)$} \label{step3}
				\State Set $E = E \cup \{(u_{i}^j, u_{0}^{j+1}), (u_{i}^j, 
				v_{0}^{j+1})\}$ and $E = E \cup \{(v_{i}^j, u_{i}^{j+1}), 
				(v_{i}^{j}, v_{i}^{j+1})\}$
				\ElsIf {$\neg x_{j} \in c_i(x)$} \label{step4}
				\State Set $E = E \cup \{(u_{i}^j, u_{i}^{j+1}), 
				(u_{i}^{j}, v_{i}^{j+1})\}$ and $E = E \cup \{(v_{i}^j, u_{0}^{j+1}), (v_{i}^j, 
				v_{0}^{j+1})\}$
				\Else
				\State Set $E = E \cup \{(u_{i}^j, u_{i}^{j+1}), 
				(u_{i}^{j}, v_{i}^{j+1})\}$ and $E = E \cup \{(v_{i}^j, u_{i}^{j+1}), 
				(v_{i}^{j}, v_{i}^{j+1})\}$
				\EndIf
				\EndFor
				\State Remove $u_i^n$ and $v_i^n$ from $G$ \label{step1} 
				\label{step_remove}
				\If {$x_{n} \in c_i(x)$} \label{step_final1}
				Set $E = E \cup \{(u_{i}^{n-1}, u_{0}^{n}), 
				(v_{i}^{n-1}, u_{0}^{n})\}$ 
				\EndIf
				\If {$\neg x_{n} \in c_i(x)$} \label{step_final2}
				Set $E = E \cup \{(u_{i}^{n-1}, v_{0}^{n}), 
				(v_{i}^{n-1}, v_{0}^{n})\}$ 
				\EndIf
				\EndFor
				\State Create vertex-label function $\ell_V(.)$ such that $\forall i \in \{0, 1, ..., m\}, j \in \{1, ..., n\}(\ell_V(v_i^j) = 0)$, $\ell_V(.)$ outputs $1$ otherwise.		
				\State \textbf{Return} $G$ and $\ell_V$ as result.
			\end{algorithmic}
		\end{algorithm}
	
	To better illustrate Algorithm~\ref{alg_reduction}, 
	Figure~\ref{fig_sat_reduction} shows an example of its output for a given 
	SAT instance. Note that, Algorithm~\ref{alg_reduction} does not consider 
	tautologies such as the empty clause or one containing $x \vee \neg 
	x$.

	\begin{figure}
		\centering
		\begin{tikzpicture}[scale = 0.6]

\fill[black] (0,6) circle (0.05cm) node[anchor = south] { \large{$\text{root}$}};
\fill[black] (-1.7,4) circle (0.05cm);
\fill[black] (1.7,4) circle (0.05cm);

\draw[-to] (0, 6)  --  (-1.7, 4);
\draw[draw,-to] (0, 6)  --  (1.7, 4);

\draw[draw,-to] (-1.7, 4)  --  (-3.5, 2);
\draw[draw,-to] (-1.7, 4)  --  (0, 2);
\draw[draw,-to] (1.7, 4)  --  (3.5, 2);

\fill[black] (-3.5,2) circle (0.05cm) node[anchor = south] {\large{$c_1$}};
\fill[black] (0,2) circle (0.05cm) node[anchor = south] { \large{$c_2$}};
\fill[black] (3.5,2) circle (0.05cm) node[anchor = south] { \large{$c_3$}};

\draw[draw,-to] (-3.5, 2)  --  (-7, 0);
\draw[draw,-to] (-3.5, 2)  --  (-5, 0);

\draw[draw,-to] (0, 2)  --  (-3, 0);
\draw[draw,-to] (0, 2)  --  (-1, 0);

\draw[draw,-to] (3.5, 2)  --  (1, 0);
\draw[draw,-to] (3.5, 2)  --  (3, 0);

\fill[black] (-7,0) circle (0.05cm) node[anchor = east] {\large{$u_1^1$}};
\fill[black] (-5,0) circle (0.05cm) node[anchor = south] {\large{$v_1^1$}};
\fill[black] (-3,0) circle (0.05cm) node[anchor = south] {\large{$u_2^1$}};
\fill[black] (-1,0) circle (0.05cm) node[anchor = south] {\large{$v_2^1$}};
\fill[black] (1,0) circle (0.05cm) node[anchor = south] {\large{$u_3^1$}};
\fill[black] (3,0) circle (0.05cm) node[anchor = south] {\large{$v_3^1$}};

\draw[draw,-to] (-5, 0)  --  (-7, -2);
\draw[draw,-to] (-5, 0)  --  (-5, -2);

\draw[draw,-to] (-3, 0)  --  (-3, -2);
\draw[draw,-to] (-3, 0)  --  (-1, -2);
\draw[draw,-to] (-1, 0)  --  (-3, -2);
\draw[draw,-to] (-1, 0)  --  (-1, -2);
\draw[draw,-to] (1, 0)  --  (1, -2);
\draw[draw,-to] (1, 0)  --  (3, -2);

\fill[black] (-7,-2) circle (0.05cm) node[anchor = east] {\large{$u_1^2$}};
\fill[black] (-5,-2) circle (0.05cm) node[anchor = west] {\large{$v_1^2$}};
\fill[black] (-3,-2) circle (0.05cm) node[anchor = east] {\large{$u_2^2$}};
\fill[black] (-1,-2) circle (0.05cm) node[anchor = west] {\large{$v_2^2$}};
\fill[black] (1,-2) circle (0.05cm) node[anchor = east] {\large{$u_3^2$}};
\fill[black] (3,-2) circle (0.05cm) node[anchor = south] {\large{$v_3^2$}};
\fill[black] (5,-2) circle (0.05cm) node[anchor = east] {\large{$u_0^2$}};
\fill[black] (7,-2) circle (0.05cm) node[anchor = west] {\large{$v_0^2$}};

\draw[draw,-to] (5, -2)  --  (5, -4);
\draw[draw,-to] (5, -2)  --  (7, -4);
\draw[draw,-to] (7, -2)  --  (5, -4);
\draw[draw,-to] (7, -2)  --  (7, -4);
\draw[draw,-to,dashed] (-1, -2)  --  (-1, -4);
\draw[draw,-to] (3, -2)  --  (5, -4);
\draw[draw,-to] (3, -2)  --  (7, -4);

\draw[draw,-to] (1,-2) .. controls (2,-4) and (5,-6) .. (7,-4);

\fill[black] (5,-4) circle (0.05cm) node[anchor = north] {\large{$u_0^3$}};
\fill[black] (7,-4) circle (0.05cm) node[anchor = north] {\large{$v_0^3$}};

\fill[black] (-7,3) circle (0.05cm) node[anchor = south] {\large{$u_0^2$}};
\fill[black] (-5,3) circle (0.05cm) node[anchor = south] {\large{$v_0^2$}};


\fill[black] (-7,-4) circle (0.05cm) node[anchor = north] {\large{$u_0^3$}};
\fill[black] (-5,-4) circle (0.05cm) node[anchor = north] {\large{$v_0^3$}};

\fill[black] (-3,-4) circle (0.05cm) node[anchor = north] {\large{$u_0^3$}};
\fill[black] (-1,-4) circle (0.05cm) node[anchor = north] {\large{$v_0^3$}};

\draw[draw,-to, dashed] (-7, 0)  --  (-7, 3);
\draw[draw,-to, dashed] (-7, 0)  --  (-5, 3);

\draw[draw,-to] (3, 0)  --  (5, -2);
\draw[draw,-to] (3, 0)  --  (7, -2);

\draw[draw,-to, dashed] (-5, -2)  --  (-7, -4);
\draw[draw,-to, dashed] (-5, -2)  --  (-5, -4);

\draw[draw,-to, dashed] (-3, -2)  --  (-3, -4);
\draw[draw,-to, dashed] (-3, -2)  --  (-1, -4);

\end{tikzpicture}
		\caption{\label{fig_sat_reduction}The resulting graph when applying 
			Algorithm~\ref{alg_reduction} on input the boolean formula $(x_1 \vee \neg 
			x_2) \wedge (x_2 \vee \neg x_3) \wedge (\neg x_1 \vee \neg x_2 \vee \neg 
			x_3)$. For the sake of a good visualization some nodes have been ``cloned'', 
			however, they actually represent a single node in the graph. Such ``cloned'' 
			nodes can be easily identified as the ones with dashed incident edges.}
	\end{figure}
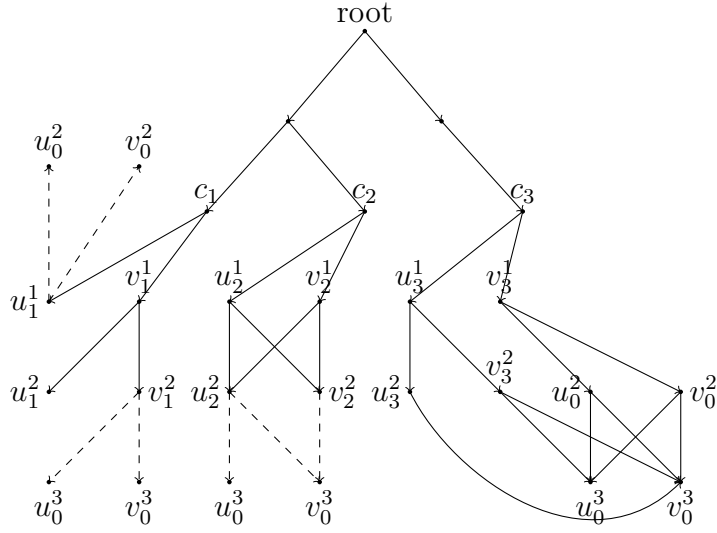
	
	\begin{lemma}\label{lemma_np_1}
		The longest walk in $G$ starting from the \emph{root} vertex has length 
		$n+\lceil \log m \rceil$ and ends either at $u_0^n$ or $v_0^n$.
	\end{lemma}
	
	\begin{proof}
		
		Let $w = w_0...w_k$ be a walk in $G = (V, E)$ starting from the \emph{root} 
		vertex. Let us assume that $w$ is maximal in the sense that $w_k$ does not 
		have out-going edges. According to Algorithm~\ref{alg_reduction}, $w_k$ does 
		not have out-going edges only if $w_k \in \{u_0^{n},v_0^{n}\}$ or $w_k \in 
		\{u_1^{n-1}, v_1^{n-1}, u_2^{n-1}, v_2^{n-1}, ..., u_m^{n-1}, v_m^{n-1}\}$ (see Step~\ref{step1} of 
		Algorithm~\ref{alg_reduction}). Therefore, either the longest walk ends at 
		$u_0^n$ or $v_0^n$ and its length is $n+\lceil \log m \rceil$ or its length is 
		$n-1+\lceil \log m \rceil$. The proof concludes by remarking that there must 
		exist at least one walk ending at $u_0^n$ or $v_0^n$ unless all the clauses 
		are empty, which is a tautology not-considered in SAT.\qed
	\end{proof}
	
	\begin{lemma}\label{lemma_np_2}
		Let $s = s_0s_1...s_{n+\lceil \log m \rceil}$ and $t = t_0t_1...t_{n+\lceil 
			\log m \rceil}$ two different maximal length walks in $G$ that start from the 
		\emph{root} vertex. Then, $\forall k \in \{0,...,n+\lceil \log m \rceil\} 
		(\ell_V(s_k) = \ell_V(t_k)) \Rightarrow \exists i\neq j (c_i \in s \wedge c_j 
		\in t)$.
	\end{lemma}
	
	\begin{proof}
		According to Algorithm~\ref{alg_reduction}, there exist $i,j \in \{1, ..., 
		m\}$ such that $c_i = s_{\lceil \log m \rceil}$ and $c_j = t_{\lceil \log m 
			\rceil}$. In addition, every vertex $s_k$ (resp. $t_k$) where $\lceil \log m 
		\rceil < k < n+\lceil \log m \rceil$ is either $u_i^{k-\lceil \log m \rceil}$ 
		(resp. $u_j^{k-\lceil \log m \rceil}$) or $v_i^{k-\lceil \log m \rceil}$ 
		(resp. $v_j^{k-\lceil \log m \rceil}$).
		
		Now, according to the vertex-label function $\ell_V$, if $\forall k \in 
		\{0,...,n-1+\lceil \log m \rceil\} (\ell_V(s_k) = \ell_V(t_k))$, then $\forall 
		k \in \{\lceil \log m \rceil+1, ..., n-1+\lceil \log m \rceil\} (s_k = 
		v_i^{k-\lceil \log m \rceil} \Leftrightarrow t_k = v_j^{k-\lceil \log m 
			\rceil})$. Similarly, if $\ell_V(s_{n+\lceil \log m \rceil}) = 
		\ell_V(t_{n+\lceil \log m \rceil})$ then $s_{n+\lceil \log m \rceil} = v_0^{n} 
		\Leftrightarrow t_{n+\lceil \log m \rceil} = v_0^{n}$ (see Lemma~\ref{lemma_np_1}). Therefore, $i = j 
		\Rightarrow s = t$, which is a contradiction. \qed
	\end{proof}
	
	\begin{theorem}\label{theo_np}
		The Binary MFS problem is NP-Hard.
	\end{theorem}
	
	\begin{proof}
		Let $\prod$ be an instance of the SAT problem and let $G$ be the graph obtained by applying Algorithm~\ref{alg_reduction} on input $\prod$. Let $\prod'$ be the problem of 
		finding the most frequent sequence of length $n+1+\lceil \log m \rceil$ in $G$ 
		starting from the vertex root. Given a solution $s$ for $\prod'$, our aim is to prove that a true 
		assignment for $\prod$ exists (and can be found) in polynomial time if and 
		only if $s$ appears $m$ times in $G$. Doing so, $\prod$ is proven to be polynomially reducible to $\prod'$, which is a Binary MFS problem.
		
		First, let us assume that the most frequent sequence $s = s_0s_1...s_{n+\lceil 
			\log m \rceil}$ in $G$ occurs exactly $m$ times. Let $w_0w_1...w_{n+\lceil 
			\log m \rceil}$ be a walk such that $\forall i \in \{0,...,n+\lceil \log m 
		\rceil\} (\ell_V(w_i) = s_i)$. By Algorithm~\ref{alg_reduction}, there must 
		exist $i \in \{1, ..., m\}$ such that $w_{\lceil \log m \rceil} = c_i$. Let $w_{k+\lceil \log m \rceil}$ be the vertex such that $w_{k+\lceil \log m 
			\rceil} \notin \{u_0^{k}, v_0^{k}\}$ and $w_{k+1+\lceil \log m \rceil} \in 
		\{u_0^{k+1}, v_0^{k+1}\}$. Note that, such a vertex exists due to 
		Lemma~\ref{lemma_np_1} and Algorithm~\ref{alg_reduction}. According to such an algorithm, either 
		$w_{k+\lceil \log m \rceil} = u_i^{k}$ and $c_i(x)$ contains the literal 
		$x_{k}$ or $w_{k+\lceil \log m \rceil} = v_i^{k}$ and $c_i(x)$ contains the 
		literal $\neg x_{k}$. Therefore, if $w_{k+\lceil \log m \rceil} = u_i^{k}$ 
		then $x_{k} = s_{k+\lceil \log m \rceil} = \ell_V(u_i^{k}) = 1$ satisfies the 
		clause $c_i(x)$, otherwise $x_{k} = s_{k+\lceil \log m \rceil} = 
		\ell_V(v_i^{k}) = 0$ does. Consequently, the assignment $x_{j} = s_{j+\lceil \log m \rceil}\ 
		\forall j \in \{1, ..., n\}$ satisfies $c_i(x)$.
		
		Considering that $s$ appears $m$ times, then by Lemma~\ref{lemma_np_2} we can 
		conclude that all the clauses are satisfied by such assignment,  whereupon we 
		finish the first part of this proof.
		
		Now, let $x = (y_1,...,y_n)$  be a true assignment for $\prod$. Let us 
		consider the induced sub-graph $G_i$ formed by the vertex $c_i$ and all the 
		other vertices reachable from $c_i$. By design of Algorithm~\ref{alg_reduction}, there exists a walk $w = 
		w_0w_1...w_{n-1}$ in 
		$G_i$ such that $\forall k \in \{1,..., n-1\}(y_k = \ell_V(w_k))$ and $w_0 = 
		c_i$. In addition, if $y_j$ satisfies clause $c_i(x)$ (\emph{i.e.}, $y_j = 1 
		\wedge x_j \in 
		c_i(x)$ or $y_j = 0 \wedge \neg x_j \in c_i(x)$), then according to 
		Algorithm~\ref{alg_reduction} either ($x_j \in 
		c_i(x) \wedge \{(u_i^{j}, u_0^{j+1}), (u_i^{j}, v_0^{j+1})\} \in E)$ or $(\neg 
		x_j \in 
		c_i(x) \wedge \{(v_i^{j}, u_0^{j+1}), (v_i^{j}, v_0^{j+1})\} \in E)$. 
		Consequently, it must hold that $w_{n-1} \in \{u_0^{n-1}, 
		v_0^{n-1}\}$ if and only if $(x_1,...,x_{n-1}) = (y_1,...,y_{n-1})$ satisfies 
		$c_i(x)$. If $(y_1,...,y_{n-1})$ does satisfies $c_i(x)$, the walk $w = 
		w_0w_1...w_{n-1}w_n$ where $w_n = u_0^{n}$ if $y_n = 1$ and $w_n = v_0^{n}$ if $y_n = 0$ holds that $\forall k \in \{1,..., n\}(y_k = \ell_V(w_k))$. On the other hand, if $(y_1,...,y_{n-1})$ does not satisfies $c_i(x)$, then $y_n = 1 \Rightarrow x_n \in c_i(x)$ and $y_n = 0 \Rightarrow \neg x_n \in c_i(x)$. According to Steps~\ref{step_final1} and~\ref{step_final2} of Algorithm~\ref{alg_reduction}, $w = 
		w_0w_1...w_{n-1}w_n$ where $w_n = u_0^{n}$ if $y_n = 1$ and $w_n = v_0^{n}$ if $y_n = 0$ is a walk holding that $\forall k \in \{1,..., n\}(y_k = \ell_V(w_k))$. As a conclusion, for every $i \in \{1, ..., m\}$ there exists a walk passing through $c_i$ that generates the sequence $\underbrace{11...11}_{\lceil \log m \rceil+1}y_1...y_n$. This result together with Lemma~\ref{lemma_np_2} conclude that such a sequence repeats exactly $m$ times.\qed
	\end{proof}

	\begin{corollary}
		The MFS problem is NP-Hard.
	\end{corollary}

	\section{Distance fraud analysis for the tree-based approach}\label{sec_algorithm}
	
	In this section, the problem of computing the distance fraud resistance of the tree-based DB protocol~\cite{DBLP:conf/isw/AvoineT09} is addressed. A naive algorithm to solve this problem consists on analyzing all the labeling functions for a full binary tree of depth $n$ and then computing the most frequent sequence for each labeling function (see Definition~\ref{def_distance_fraud}). This results in a time-complexity of $O(2^{2^n+n})$, which is unfeasible even for small values of $n$. We propose an algorithm that reduces this time-complexity to $O(2^{2n}n)$. Although still exponential, it might be used up to reasonable values of $n$ (\emph{e.g.}, $n = 32$).
	
	For the sake of clarity, we first adapt Definition~\ref{def_distance_fraud} to the context of the tree-based proposal.
	
	\begin{problem}[Tree-based distance fraud problem]
		\label{def_tree_problem}
		Let $\prod$ be a tree-based DB protocol with $n$ time-critical sessions that uses the full binary tree $T = (V, E)$ and $root \in V$ as the starting vertex. Let $M_{n}$ be a random variable on the sample space of all labeling functions $\ell: V \rightarrow \{0,1\}$ that outputs the maximum value $\max(occ_{ root}^T(t))$ where $t \in \{0,1\}^{n+1}$. The \emph{tree-based distance fraud problem} consists on finding the expectation of the random variable $M_{n}$.
	\end{problem}
	
	\begin{theorem}\label{theo_computation}
		Let $m$ and $n$ be two positive integers. Let $T^m_n$ be a tree such that: i) the root has $2m$ children and ii) each root's children is the root of a full binary tree of depth $n-1$. Let $M^m_n$ be the random variable on the sample space of all binary vertex-labeling functions over $T^m_n$ that outputs the maximum value $\max(occ_{ root}^{T^m_n}(t))$. The expectation of the random variable $M_{n}$ can be computed as follows:
		
		\begin{equation*}
		E(M_n) = E(M^1_n) = \sum_{i = 1}^{2^n}\left( \Pr(M^1_n < i+1) -\Pr(M^1_n < i) \right)i.
		\end{equation*}
		where 
		\begin{align*}
		\Pr(M^m_n < x) = & \sum_{i = 0}^{2m}\frac{\binom{2m}{i}}{2^{2m}}\left(\Pr(M^{i}_{n-1} < x)\Pr(M^{2m-i}_{n-1} < x) \right).
		\end{align*}
		and 
		\begin{equation*}
		\Pr(M^m_n < x) = \left\{\begin{array}{ll}
		0  & \text{if} \quad x = 1\\
		0  & \text{if} \quad n = 1 \wedge m > x\\
		\frac{1}{2^{2m}}\left(
		\binom{2m}{m} + 2\sum\limits_{i = m+1}^{x-1}\binom{2m}{i}
		\right)  & \text{if} \quad n = 1 \wedge m \leq x \leq 2m\\
		1  & \text{if} \quad n = 1 \wedge 2m < x \\
		1  & \text{if} \quad m = 0\\
		\end{array}
		\right.
		\end{equation*}
		
	\end{theorem}
	
	\begin{proof}
		A full binary tree of depth $n$ can be denoted as $T^1_n$ and the random variable $M_n$ is equivalent to $M^1_n$. Therefore, in what follows, we focus on computing the expectation of the random variable $M_{n}^m$.

		Let us consider now a labeling function $\ell$ over $T^m_n$. Let $V_0$ and $V_1$ be the set of children of the $T^m_n$'s root labeled with $0$ and $1$ respectively. Let $C^0_{1}, C^0_{2}, ..., C^0_{2|V_0|}$ and $C^1_{1}, C^1_{2}, ..., 
		C^1_{2|V_1|}$
		be the subtrees rooted in the children of the vertices in $V_0$ and $V_1$ respectively. Let $X$ be a labeled 
		tree whose root is labeled with $0$ and the root's children are the full 
		binary trees $C^0_{1}, C^0_{2}, ..., C^0_{2|V_0|}$. In the same vein, 
		$Y$ is defined as a labeled tree whose root is labeled with 
		$1$ and the root's children are the full binary trees $C^1_{1}, C^1_{2}, ..., 
		C^1_{2|V_1|}$. It can be noted that, if a sequence $t = t_1...t_n$ occurs 
		exactly $k$ times either in $X$ or $Y$, then the sequence $t = t_0t_1...t_n$ where $t_0$ is the label of $T_n^m$'s root also appears exactly $k$ times in $T^m_n$. Therefore, taking into account that $X = T^{|V_0|}_{n-1}$ and $Y = T^{|V_1|}_{n-1}$, the following recurrent result can be obtained:

		\begin{align}\label{eq_tree_1}
		\Pr(M^m_n < x) = & \sum_{i = 0}^{2m}\Pr(|V_0| = i)\left(\Pr(M^{i}_{n-1} < x)\Pr(M^{2m-i}_{n-1} < x) \right)\notag \\
		= & \sum_{i = 0}^{2m}\frac{\binom{2m}{i}}{2^{2m}}\left(\Pr(M^{i}_{n-1} < x)\Pr(M^{2m-i}_{n-1} < x) \right).
		\end{align}
		
		Equation~\ref{eq_tree_1} shows that $\Pr(M^m_n < x)$ could be computed recursively. To do so, stop conditions must be found as follows:
		
		\begin{equation}\label{eq_tree_2}
			\Pr(M^m_n < x) = \left\{\begin{array}{ll}
				0  & \text{if} \quad x = 1\\
				0  & \text{if} \quad n = 1 \wedge m > x\\
				\frac{1}{2^{2m}}\left(
				\binom{2m}{m} + 2\sum\limits_{i = m+1}^{x-1}\binom{2m}{i}
				\right)  & \text{if} \quad n = 1 \wedge m \leq x \leq 2m\\
				1  & \text{if} \quad n = 1 \wedge 2m < x \\
				1  & \text{if} \quad m = 0\\
			\end{array}
			\right.
		\end{equation}
		
		Let us analyze $\Pr(M^m_1 < x)$, which is the less trivial stop condition in Equation~\ref{eq_tree_2}. Since $T_1^m$ has depth $1$ and $2m$ children,  $T_1^m$ generates $2m$ sequences, $p$ of them ending with $0$ and $q$ with $1$. Consequently, $M^m_1 = \max(p, q) \geq m$ and thus, $\Pr(M^m_1 < x) = 0$ if $x < m$. Similarly $M^m_1 \leq 2m$, which implies that $\Pr(M^m_1 < x) = 1$ if $x > 2m$. Finally, let us assume that $m \geq x \geq 2m$. In this case, $M^m_1 < x$ holds if $M^m_1 \in \{m, m+1, ..., x-1\}$, therefore, $\Pr(M^m_1 < x) = \sum\limits_{i = m}^{x-1}\Pr(M^m_1 = i)$ where $\Pr(M^m_1 = i) = \frac{\binom{2m}{m}}{2^{2m}}$ if $i = m$, otherwise  $\Pr(M^m_1 = i) = 2\frac{\binom{2m}{i}}{2^{2m}}$. This yields to $\Pr(M^m_1 < x) = \frac{1}{2^{2m}}\left(
		\binom{2m}{m} + 2\sum\limits_{i = m+1}^{x-1}\binom{2m}{i}
		\right)$ if $m \geq x \geq 2m$.
		
		The proof concludes by using the definition of expectation for a discrete variable together with Equations~\ref{eq_tree_1} and~\ref{eq_tree_2}.\qed
		
	\end{proof}
	
	\noindent \textbf{Time-complexity analysis.}	The result provided by Theorem~\ref{theo_computation} can be implemented by a dynamic algorithm, meaning that a three-dimensional matrix will dynamically store the values of $\Pr(M^m_n < x)$ and will use them when needed. In the worst case, the algorithm requires to fill the whole matrix, which results in a time-complexity of $O(2^n\times n \times 2^n) = O(2^{2n}n)$.

	\section{Discussion and Conclusions}\label{sec_conclusions}
	
	Before the introduction of graph-based DB protocols, computing resistance 
	to distance fraud was not a big issue (\emph{e.g.}, Hancke and 
	Kuhn~\cite{Hancke:2005:RDB:1128018.1128472} and Kim and 
	Avoine~\cite{KimA-2011-ieeetwc} proposals.) Actually, the well-known 
	early-reply strategy used to analyze distance fraud security implicitly 
	assumes that the adversary is able to compute the \emph{best} answer 
	without knowing the challenges and within a ``reasonable'' time frame. In 
	this article, however, we have shown that this assumption might not hold 
	for graph-based DB protocols by proving that the Binary MFS problem is 
	NP-Hard. This opens two interesting research questions: i) What instances 
	of the Binary MFS problem are actually hard to solve? ii) In a practical 
	setting, does the adversary have enough time to solve a probably 
	exponential problem between the end of the slow phase and the beginning of 
	the fast phase? iii) What kinds of heuristics can be used and what would 
	the implications be?
	
	Even though we do not give answers to those questions, we provide a clue of how to build graph-based DB protocols resistant to distance fraud. As indicate our reduction from the SAT problem, a good strategy is to label the vertices of $G$ as follows: if the vertices $u$ and $v$ have incident edges from the same vertex, then $\forall b \in \{0, 1\} (\ell_V(u) = b  \Leftrightarrow \ell_V(v) = b \oplus 1)$. Doing so, $G$ is likely to generate all the sequences $\{0,1\}^n$ just once, in which case its resistance to distance fraud achieves the lower  bound $1/2^n$. As a consequence, we conjecture that the best graph DB protocol in terms of mafia fraud constrained to have no more than certain number of nodes, is also the best in terms of distance fraud. Note that, this conjecture becomes trivial if no limit on the size of the graph is considered.
	
	This article has also addressed the problem of computing the distance fraud security of the tree-based DB proposal~\cite{DBLP:conf/isw/AvoineT09}. This is an inherent exponential problem since a graph with $N$ nodes can be labeled in $2^N$ different ways. The tree-based proposal uses a tree with $2^n$ nodes and thus $2^{2^n}$ labelling functions exists. However, we provide an algorithm that avoids considering all the labelling functions and has a time complexity of $O(2^{2n}n)$, which is significantly better than the naive approach with $O(2^{2^n+n})$. This result makes realistic the challenge of computing the distance fraud security of the two graph-based DB protocols proposed up-to-date; the tree-based approach~\cite{DBLP:conf/isw/AvoineT09} by using the proposed algorithm ($O(2^{2n}n)$), and the Poulidor protocol~\cite{Trujillo-Rasua:2010:PDP:1926325.1926352} by simply using a brute-force algorithm ($O(2^{3n})$). Doing so, both can be fairly compared with other state-of-the-art DB protocols. Such a challenge is out of the scope of this article and is left as future work, though.

\subsubsection*{Acknowledgments.} The author thanks to Gildas Avoine, Sjouke 
Mauw, Juan Alberto 
Rodriguez-Velazquez, and Alejandro Estrada-Moreno for their invaluable 
comments and feedback.

	\end{document}